\documentclass[
final
 , nomarks
]{dmtcs-episciences}

\usepackage{graphicx}
\usepackage{mathptmx}
\usepackage{latexsym}
\usepackage{amssymb}
\usepackage{amstext}
\usepackage{algpseudocode}
\usepackage{algorithm}

\newtheorem{thm}{Theorem}
\newtheorem{theorem}[thm]{Theorem}

\newcommand{\Real}{\mathbb{R}}
\newtheorem{obs}{Observation}

\author[Hossein Boomari and Mojtaba Ostovari and Alireza Zarei]{
Hossein Boomari\affiliationmark{1}
\and 
Mojtaba Ostovari\affiliationmark{1}
\and 
Alireza Zarei\affiliationmark{1}
}
\title[]{Recognizing Visibility Graphs of Polygons with Holes and Internal-External Visibility Graphs of Polygons}
\affiliation{
  Department of Mathematical Sciences\\Sharif University of Technology
}
\begin{document}
% This is only used if you are compiling for a volume before vol 25
% \publicationdetails{VOL}{2015}{ISS}{NUM}{SUBM}
% This is the new form of collecting the data, starting with vol 25
% \publicationdata
% {vol. 25:3 special issue for main purpose}
% {2022}
% {1}
% {10.46298/dmtcs.10472}
%{1998-10-14; 1998-10-14; 2002-07-19; 2014-02-05; 2015-09-09; 2022-12-25}
%{2022-12-3}
% {2022-12-3; None}
% {2023-1-1}
\maketitle
\begin{abstract}
The visibility graph of a polygon corresponds to its internal diagonals and boundary edges. For each vertex on the boundary of the polygon, we have a vertex in this graph and if two vertices of the polygon see each other, there is an edge between their corresponding vertices in the graph. Two vertices of a polygon see each other if and only if their connecting line segment completely lies inside the polygon, and they are externally visible if and only if this line segment completely lies outside the polygon. Recognizing visibility graphs is the problem of deciding whether there is a simple polygon whose visibility graph is isomorphic to a given input graph. This problem is well-known and well-studied, but yet widely open in geometric graphs and computational geometry. 

Existential Theory of the Reals is the complexity class of problems that can be reduced to the problem of deciding whether there exists a solution to a quantifier-free formula $F(X_1,X_2,...,X_n)$, involving equality and inequality of real polynomials with real variables. The complete problems for this complexity class are called $\exists \Real-Complete$.

In this paper, we show that recognizing visibility graphs of polygons with holes is $\exists \Real-Complete$; recognizing visibility graphs of simple polygons is $\exists \Real-Complete$ when we have both internal and external visibility graphs; recognizing visibility graphs of simple polygons is $\exists \Real-Complete$ when we have both internal and external visibility graphs and a \textit{BlockingVertexAssignment}.
\end{abstract}

\renewcommand\thefootnote{}

\renewcommand\thefootnote{\fnsymbol{footnote}}
\setcounter{footnote}{1}

\section{Introduction}
\label{intro}
The visibility graph of a simple planar polygon is a graph in which there is a vertex for each vertex of the polygon and for each pair of visible vertices of the polygon there is an edge between their corresponding vertices in this graph. Two points in a simple polygon are internally (resp. externally) visible from each other when their connecting segment completely lies inside (resp.  outside) the polygon. In this definition, each pair of adjacent vertices on the boundary of the polygon is assumed to be visible from each other. This implies that we always have a Hamiltonian cycle in a visibility graph which determines the order of vertices on the boundary of the underlying polygon. A polygon with holes has some non-intersecting holes inside its boundary. In these polygons, the area inside a hole is considered as the outside area. Then, internal and external visibility graphs of such polygons are defined in the same way as defined for simple polygons. In the visibility graph of a polygon with holes, we have the sequence of vertices corresponding to the boundary of each hole, as well. 

Computing the visibility graph of a given simple polygon has many applications in computer graphics~\cite{CG}, computational geometry~\cite{ghosh-book} and robotics~\cite{robot}. There are several efficient polynomial time algorithms for this problem~\cite{ghosh-book}.

This concept has been considered in reverse as well: Is there any simple polygon whose visibility graph is isomorphic to a given graph, and, if there is such a polygon, is there any way to reconstruct it (finding positions for its vertices in the plane)? The former problem is known as recognizing visibility graphs and the latter one is known as reconstructing a polygon from a given visibility graph. The computational complexities of both these problems are widely open. The only known result about the computational complexity of these problems is that they belong to \textit{PSPACE}~\cite{everet-thesis} complexity class. More precisely, they belong to the class of \textit{Existential theory of the reals}~\cite{exist}. This means that it is not even known whether these problems are in \textit{NP} or can be solved in polynomial time. Even, if we are given the Hamiltonian cycle of the visibility graph which determines the order of vertices on the boundary of the target polygon, the exact complexity class of these problems is still unknown.

However, these problems have been solved efficiently only for special cases of {\it spiral}, {\it tower}, and {\it anchor} polygons. These results are obtained by Everett and Corneil~\cite{spiral} for spiral polygons, by Colley~{\it et al.}~\cite{tower} for tower polygons, and by Boomari and Zarei~\cite{DBLP:journals/jgaa/BoomariZ22} for anchor polygons. In a spiral polygon, there is at most one concave chain, the boundary of a tower polygon is composed of two concave chains and a single edge, and an anchor polygon is a tower polygon whose base edge is a convex chain. The recognizing and reconstruction problems have been solved for these special cases in linear time in terms of the size of the graph.

Although there is some progress in recognizing and reconstructing problems, there have been plenty of studies on characterizing visibility graphs. In 1988, Ghosh introduced three necessary conditions for visibility graphs and conjectured their sufficiency~\cite{ghosh3}. In 1990, Everett proposed a graph that rejects Ghosh's conjecture~\cite{everet-thesis}. She also refined Ghosh's third necessary condition to a new stronger one~\cite{ghoshn}. In 1992, Abello \textit{et al}. built a graph satisfying Ghosh's conditions and the stronger version of the third condition which was not the visibility graph of any simple polygon~\cite{counter3}, disproving the sufficiency of these conditions. In 1997, Ghosh added his fourth necessary condition and conjectured that this condition along with his first two conditions and the stronger version of the third condition are sufficient for a graph to be a visibility graph. Finally, in 2005 Streinu proposed a counter-example for this conjecture~\cite{counter5}. Independently, in 1994, Abello \textit{et al.} introduced the notion of \textit{q-persistant graphs} and \textit{BlockingVertexAssignment}. The \textit{BlockingVertexAssignment} is a proper function from non-visible pairs to their blocking vertices that satisfies four conditions. They proved that each visibility graph has at least one \textit{BlockingVertexAssignment}. They conjectured that these constraints are verifiable efficiently~\cite{matroid}. However, there is still no polynomial time algorithm for finding or verifying the existence of a \textit{BlockingVertexAssignment}, and their complexity classes are open. Later in 1995,  Abello \textit{et al.} added one more constraint to these constraints and proved their sufficiency for recognition and reconstruction of 2-spiral polygons\footnote{polygons with at most 2 concave chains} and proposed an efficient method for recognition and reconstruction problems for such polygons when we have the Hamiltonian cycle and a \textit{BlockingVertexAssignment} as input~\cite{twospiral}. But by now, there is no efficient algorithm for obtaining a \textit{BlockingVertexAssignment} for 2-spiral polygons from which the recognizing and reconstruction problems could be solved efficiently.

Existential theory of the reals ($\exists \Real$) is a complexity class that was implicitly introduced in 1989~\cite{npnpc} and explicitly defined by Shor in 1991~\cite{shor1991}. In 2009~\cite{schaefer}, Schaefer used the notation $\exists \Real$ for this class and showed that it is the complexity class of the problems which can be reduced to the problem of deciding, whether there is a solution for a Boolean formula $\phi:\{True,False\}^n\rightarrow \{True,False\}$ in propositional logic, in the form $\phi(F_1(X_1,X_2,...,X_N) , F_2(X_1,X_2,...,X_N), ..., F_n(X_1,X_2,...,X_N))$,
where each $F_i:\Real^N \rightarrow \{True, False\}$ consists of a polynomial function $G_i:\Real^N:\rightarrow \Real$ on some real variables, compared to $0$ with one of the comparison operators in $\{<,\leq,=,>,\geq\}$ (for example $G_i(X_1,X_2)=X_1^3X_2^2-X_1X_2^3$ and $F_i(X_1,X_2) \equiv G_i(X_1,X_2)<0$). Equivalently, it is the complexity class of the problem of deciding the emptiness of a semialgebraic set~\cite{exist}. Satisfiability of quantifier-free Boolean formula belongs to $\exists \Real$. Therefore, $\exists \Real$ includes all $NP$ problems. In addition, $\exists \Real$ strictly belongs to $PSPACE$~\cite{exist} and we have $NP \subseteq \exists \Real \subset PSPACE$. Many other decision problems, especially geometric problems, belong to $\exists \Real$ and some are complete for this complexity class~\cite{artgallery,ERP,embedding,space,multinash,quantom,planar,PVG}. Recognizing \textit{LineArrangement} (\textit{Stretchability}), simple order type, intersection graph of segments, intersection graph of unit disks, and visibility graph of a point set in the plane are some problems which are complete for $\exists \Real$ or simply $\exists \Real-Complete$~\cite{ERP,PVG}. We discuss recognizing \textit{LineArrangement} (\textit{Stretchability}) problem in more detail in this paper in Section~\ref{sec:B}.
In the most recent related result, in 2020 we showed that recognizing the visibility graph of triangulated irregular networks is $\exists \Real-Complete$~\cite{TIN}.

In this paper, we show that recognizing the visibility graph of a polygon with holes is $\exists \Real \text{-}Complete$. However, this result has been obtained historically at the same time by Hoffmann and Merckx\cite{hoffmann} using a completely different method\footnote{We first submitted this result in Dec-2017 to SOCG2018 and then published it on arXiv in Apr-2018\cite{ours}. However, their result which was published on arxiv in January 2018, has not been published as a peer-reviewed document anywhere, yet. In addition to the completely different approaches, we have extended our result to the simple polygons( without holes) when the external visibility graph is known as well. Moreover, in our proof, the problem is still $\exists\Real \text{-}Completeness$ even if we have the correct order of the vertices on the boundary of the target polygon.}. They first prove the $\exists\Real \text{-}Completeness$ of recognizing the \textit{AllowableSequences} and then reduce this problem to recognizing the visibility graphs of a polygon with holes. Compared to this result our method uses a simpler approach towards proving the $\exists\Real \text{-}Completeness$ of realizing the visibility graph of polygons with holes. In addition, we prove the $\exists\Real \text{-}Completeness$ of realizing the Internal-External visibility graph of simple polygons as well. Precisely, we show that recognizing the visibility graph of a simple polygon is also $\exists\Real \text{-}Complete$ when we have both its internal and external visibility graphs as input. We also prove that recognizing the visibility graph problem for simple polygons is still $\exists \Real\text{-}Complete$ when we have both internal and external visibility graphs and a \textit{BlockingVertexAssignment} function as input.
Based on these results, we conjecture that the main open problem, recognizing visibility graphs of simple polygons, is also $\exists \Real\text{-}Complete$.
\section{Preliminaries and Definitions}
\label{sec:B}
In this section, we give a brief survey on Ghosh's necessary conditions for a visibility graph and describe the problems of recognizing \textit{LineArrangement} and \textit{Stretchability} in the plane. We need these details in some parts of our proofs. At the end of this section, we introduce the required definitions and basic facts and observations to be used in the rest of the paper.

\subsection{Ghosh's necessary conditions}
As stated before, there are 4 necessary but not sufficient conditions that a graph must have to be the visibility graph of a simple polygon. These conditions are defined on the input visibility graph and a Hamiltonian cycle which is assumed to be the order of vertices on the boundary of the target polygon. We review the first two conditions here, briefly.

\subsubsection{First necessary condition}
Each ordered cycle of the visibility graph of length more than three has some chords. The order of vertices in such a cycle must follow the Hamiltonian cycle, and a chord is an edge between two non-adjacent vertices of the cycle. This condition is a consequence of the fact that each simple polygon has a triangulation.

\subsubsection{Second necessary condition}
\label{cond:2}
Each non-visible pair of vertices has a blocking vertex.  A vertex $o$ in a visibility graph is a blocking vertex for a non-visible pair of vertices $(p,q)$, if all vertices between $p$ and $o$ (including $p$), on the Hamiltonian cycle, are non-visible to all vertices between $o$ and $q$ (including $q$). 

\begin{obs}
For a non-visible pair $(p,q)$, there are at least one and at most two blocking vertices that are visible from $p$. These candidates are the last visible vertices in the clockwise and counter-clockwise walks, from $p$ toward $q$ along the Hamiltonian cycle.
\end{obs}
\begin{obs}
\label{obs:2}
For a non-visible pair $(p,q)$, on a walk $W$ from $p$ to $q$ along the boundary cycle, if the first vertex that is visible from $q$ is before the last vertex that is visible from $p$, then $W$ has no $(p,q)$ blocking vertex. 
\end{obs}

\subsection{Line arrangement and stretchability}
\label{sec:stret}
For a set of lines in the plane, the problem of describing their arrangement, called \textit{LineArrangement}, is important and fundamental in combinatorics and has been considered thoroughly in computational geometry. This description for a set of lines $l_1,l_2,...,l_n$ consists of their vertical order with respect to a vertical line on the left of all their intersections, and for each line $l_i$, the order of lines that are intersected by $l_i$ when we traverse $l_i$ from left to right (we assume that none of the input lines $l_i$ is vertical). Recognizing the existence of a set of lines in the plane with the given \textit{LineArrangement} is called \textit{recognizing LineArrangement} or simply \textit{LineArrangement} problem. When the lines are in general position (all pairs of lines intersect and no 3 lines intersect at the same point) the problem is called \textit{SimpleLineArrangement}. It has been proved that \textit{SimpleLineArrangement} is $\exists \Real \text{-}Complete$~\cite{ERP}. 

A pseudo-line is a monotone curve with respect to the $X$ axis (each vertical line intersects the curve at exactly one point). Assuming that no pair of pseudo-lines intersect each other more than once, we can describe an instance of recognizing \textit{PseudoLineArrangement} problem in the same way as we did for \textit{LineArrangement}.  However, recognizing \textit{PseudoLineArrangement} belongs to the $P$ complexity class and it can be decided with a Turing machine in polynomial time~\cite{allow}. Because we need such a realization algorithm, a pseudo-code implementation of this algorithm has been given in Algorithm~\ref{alg:A}.

In this algorithm, we reconstruct each pseudo-line as a chain of line segments. This algorithm receives as input, the initial vertical order of pseudo-lines and a queue $S_i$ for each pseudo-line that contains the left to right order of intersections of this pseudo-line with other pseudo-lines. There is no duplicate member in $S_i$'s, otherwise, $l_i$ must intersect another pseudo-line $l_j$ more than once which is a contradiction and in these cases, the algorithm rejects the input. The algorithm starts with $n$ empty queues, $L_i$'s, each of which will contain the sequence of vertices of the corresponding chain of the pseudo-line $l_i$. The point $(0,i)$ is inserted into $L_i$ as the initial point of $l_i$. Then, in each step, the algorithm finds two pseudo-lines that intersect and swaps their order along an imaginary vertical sweep line that moves from left to right. This algorithm, for each intersection between a pair of pseudo-lines $l_p$ and $l_q$, adds three points to $L_p$ and $L_q$ to swap their order(See Fig.~\ref{fig:1}-b). When the algorithm cannot find a  proper pair of pseudo-lines to swap, it means that the input is not recognizable and the input is rejected. As Fig.~\ref{fig:1}-c shows, when there is more than one choice for $(p,q)$, anyone can be selected and it does not affect the rest of the algorithm.

This algorithm recognizes and reconstructs a \textit{PseudoLineArrangement} in polynomial time and obtains a set of pseudo-lines, which their break-point vertices do not necessarily correspond to their intersections. For example, point $P$ in Fig.~\ref{fig:1}-b is not an intersection between the pseudo-lines. It is simple to show that we can remove these non-intersection break-point vertices from the pseudo-lines without violating input configuration constraints. Fig.~\ref{fig:1}-d shows how removing these extra break-point vertices from the chains produces new \textit{PseudoLineArrangement} which have the same order of intersections as the input configuration.

Trivially, if an instance of the \textit{LineArrangement} problem is realizable, it has a \textit{PseudoLineArrangement} realization as well. On the other hand, if an instance of the \textit{PseudoLineArrangement} has a realization in which all segments of each pseudo-line lie on the same line, the input instance has also a \textit{LineArrangement} realization as well. 

\begin{figure}[ht]
\centerline{\includegraphics[scale=0.75]{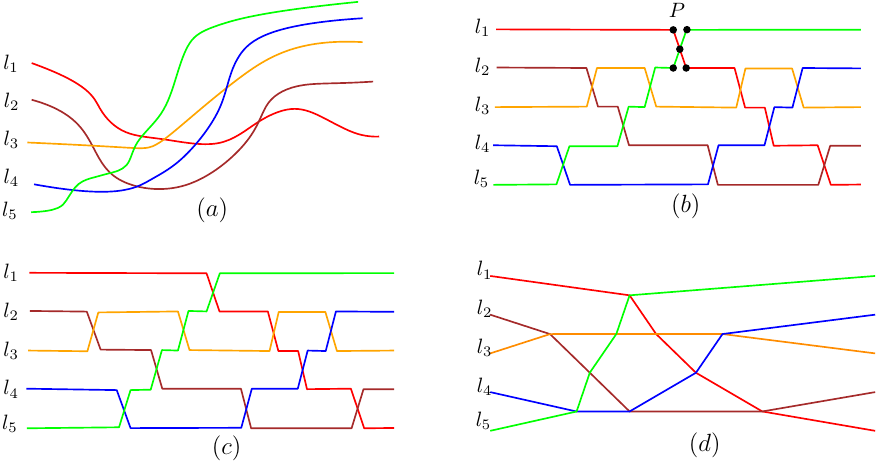}}
\caption{The reconstruction algorithm for \textit{PseodoLineArrangement}.\label{fig:1} }
\end{figure}

Therefore, we  can describe the \textit{LineArrangement} problem as follows:
\begin{itemize}
\item Is it possible to stretch a \textit{PseudoLineArrangement} of a given line arrangement description such that each pseudo-line lies on a single line? 
\end{itemize}
This problem is known as \textit{Stretchability}. As stated before, pseudo-line arrangement belongs to the $P$ complexity class and can be recognized and reconstructed efficiently(Algorithm~\ref{alg:A}). Therefore, $\exists \Real\text{-}Completeness$ of \textit{LineArrangement} implies that \textit{Stretchability} is $\exists \Real\text{-}Complete$.

%\subsection{Reconstruction algorithm to reconstruct \textit{PseudoLineArrangement}}
%\label{alg:A}
\begin{algorithm}[ht]
\centerline{\includegraphics[scale=1]{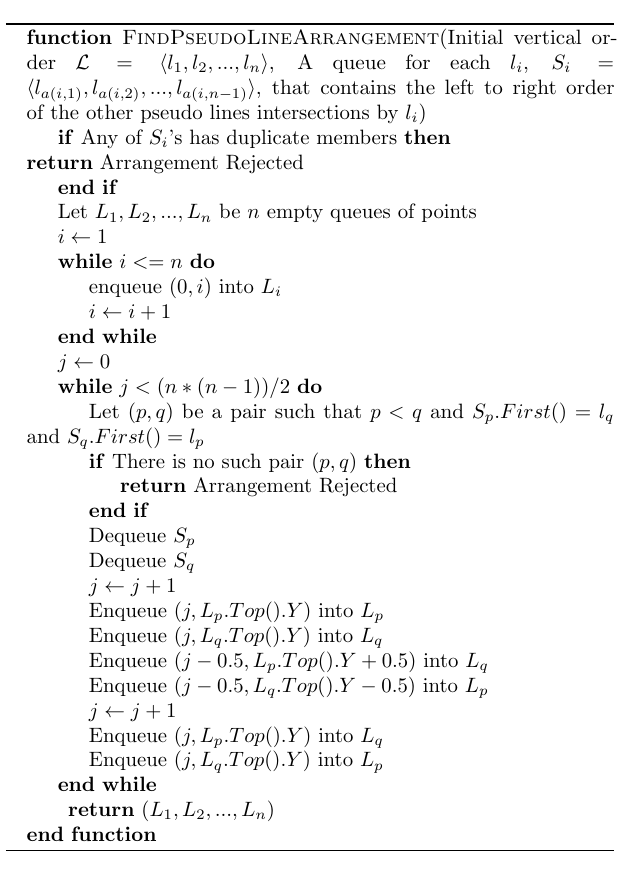}}
\caption{Recognizing and Reconstruction algorithm for \textit{PseudoLineArrangement}\label{alg:A}}
\end{algorithm}

\subsection{Visibility graph of a polygon with holes}
A polygon with holes is a simple polygon with a set of non-colliding holes (simple polygons) inside. The internal areas of the holes belong to the outside area of the polygon. In these polygons, two vertices are visible from each other if their connecting segment completely lies inside the polygon. The visibility graph of a polygon with holes is a graph whose vertices correspond to vertices of the outer boundary of the polygon and vertices of the holes, and in this graph, there is an edge between two vertices if and only if their corresponding vertices in the polygon are visible from each other (see Fig.~\ref{fig:2}). In this paper, we assume that along with the visibility graph, we have the cycles that correspond to the order of vertices on the boundary of the polygon and the holes. The cycle that corresponds to the external boundary of the polygon is called the external cycle(see Fig.~\ref{fig:2}).

\begin{figure}[ht]
\centerline{\includegraphics[scale=0.45]{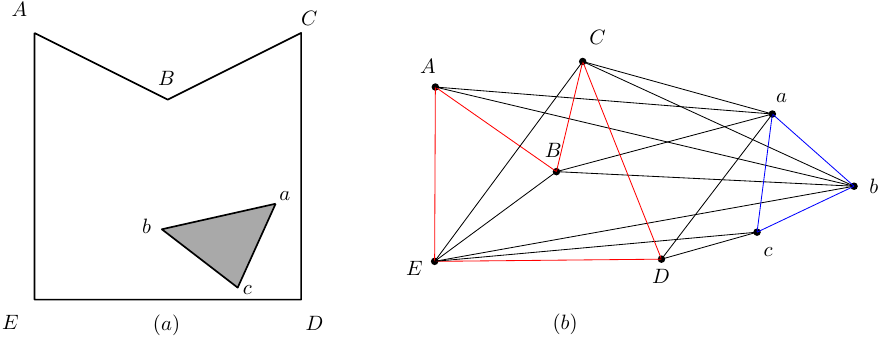}}
\caption{A polygon with one hole (a), and its visibility graph (b).\label{fig:2} }
\end{figure}

\subsection{Internal-external visibility graphs}
When we say that two points of a polygon are visible, it means that they see each other inside the polygon. However, the visibility can be defined similarly for the outside of the polygon. When we are going to talk about both of these visibilities, the former is called internal visibility and the latter is called external visibility. Precisely, two vertices are externally visible to each other if their connecting segment lies outside the polygon. Similarly, the external visibility graph of a polygon is a graph whose vertices correspond to the vertices of the polygon, and its edges correspond to the external visibility relations. Having both these graphs separately is called the internal-external visibility graphs of a polygon(see Fig.~\ref{fig:3}).

\begin{figure}[ht]
\centerline{\includegraphics[scale=0.45]{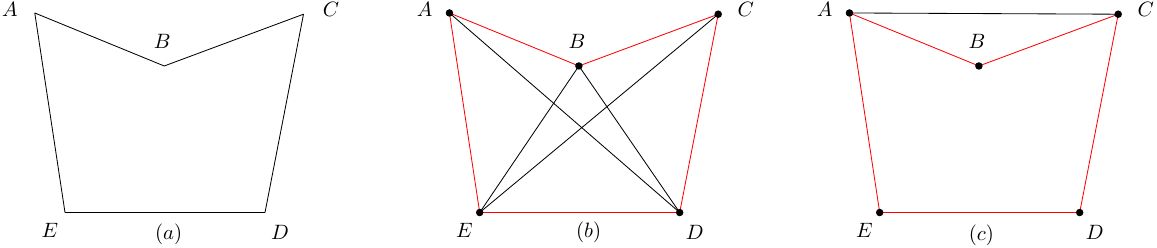}}
\caption{A polygon (a), its internal (b) and external (c) visibility graphs.\label{fig:3} }
\end{figure}

\section{Complexity of Recognizing Visibility Graphs of Polygons with Holes}
\label{sec:C}
In this section, we show that recognizing the visibility graph of a polygon with holes is $\exists \Real\text{-}Complete$. This is done by reducing an instance of the stretchability problem to an instance of this problem. 

In Section~\ref{sec:stret} we showed that we can describe the line arrangement problem as an instance of stretchability of pseudo-lines in which each pseudo-line is composed of a chain of segments and the break-point vertices of these chains(except the first and the last endpoints of the chains) correspond to the intersection points of the pseudo-lines. We build a visibility graph $\mathcal{G}$, an external cycle $\mathcal{P}$, and a set of boundary cycles $\mathcal{H}$ from an instance of such a stretchability problem, and prove that the pseudo-line arrangement is stretchable in the plane if and only if there exists a polygon with holes whose visibility graph is $\mathcal{G}$, its external cycle is $\mathcal{P}$ and the set of boundary cycles of its holes is $\mathcal{H}$.

Assume that $(\mathcal{L}, \mathcal{S})$ is an instance of the stretchability problem where, as described in Algorithm~\ref{alg:A}, $\mathcal{L}=\langle l_1,l_2,...,l_n \rangle$ is the sequence of the pseudo-lines and $\mathcal{S}=\langle S_1,S_2,...,S_n \rangle$ is the sequence of the intersections of these pseudo-lines in which  $S_i=\langle l_{a(i,1),...,a(i,n-1)}\rangle$ is the order of lines intersected by $l_i$. Let denote by $(\mathcal{G},\mathcal{P},\mathcal{H})$ the corresponding instance of the visibility graph realization in which $\mathcal{G}$ is the visibility graph, $\mathcal{P}$ is the external cycle of the outer boundary of the polygon and $\mathcal{H}=\{H_1,H_2,...,H_k\}$ is the set of boundary cycles of its holes. To build this instance, consider an example of such $(\mathcal{L},\mathcal{S})$ instance shown in Fig.~\ref{fig:4}-a. This figure shows a pseudo-line realization obtained from Algorithm~\ref{alg:A} for an instance of four pseudo-lines. If this instance is stretchable, like the one shown in Fig.~\ref{fig:4}-b, we can build a polygon with holes like the ones shown in Fig.~\ref{fig:4}-c. The outer boundary of this polygon and the boundary of its holes lie along a set of convex curves connecting the endpoints of each stretched pseudo-line. Precisely, for each stretched pseudo-line $l_i$, as in Fig.~\ref{fig:4}-b, there is a pair of convex chains on both sides that connect its endpoints. This pair of convex chains are sufficiently close to their corresponding stretched pseudo-line. For each break-point vertex of the pseudo-lines, there are four intersection points on their corresponding chains(like point $o$ in Fig.~\ref{fig:4}-c). The pair of convex chains, for each pseudo-line $l_i$, make a convex polygon which is called its channel and is denoted by $Ch(l_i)$. The outer boundary of the target polygon and the boundary of its holes are obtained by removing those segments of the chains that lie inside another channel (see Fig.~\ref{fig:4}-c). Note that, we do not have the stretched realization of $(\mathcal{L}, \mathcal{S})$ instance of the stretchability problem. But, from the pseudo-line realization, we can determine $\mathcal{G}$, $\mathcal{P}$ and $\mathcal{H}$ of the corresponding instance $(\mathcal{G},\mathcal{P},\mathcal{H})$ in polynomial time. As shown in Fig.~\ref{fig:4}-d, $\mathcal{P}$ and $\mathcal{H}$ are obtained by imaginary drawing a channel for each pseudo-line $l_i$. Finally, the vertex set of $\mathcal{G}$ is the set of all break-point vertices of these convex chains, and, two vertices are connected by an edge if and only if they belong to the boundary of the same channel. The following theorem shows the relationship between $(\mathcal{L},\mathcal{S})$ and $(\mathcal{G},\mathcal{P},\mathcal{H})$ instances.

\begin{figure}[ht]
\centerline{\includegraphics[scale=0.5]{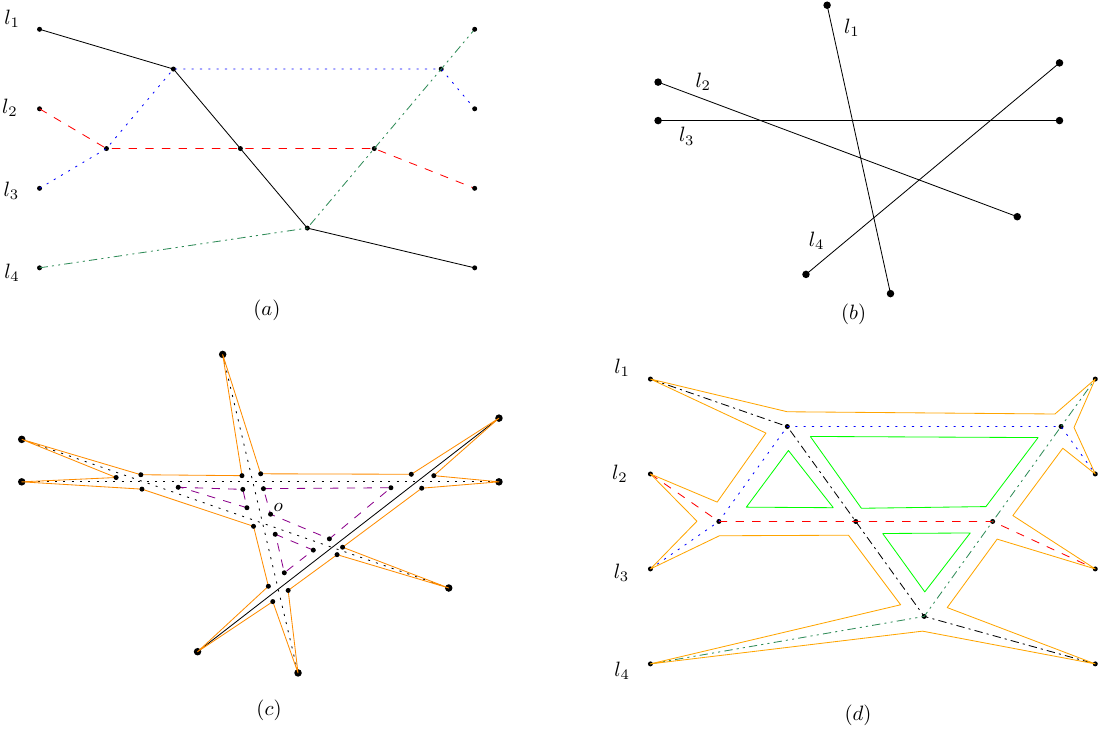}}
\caption{A polygon with holes which is constructed from an instance of the \textit{PseudoLineArrangement} problem.\label{fig:4} }
\end{figure}

\begin{theorem}
\label{thm:1_1}
An instance $(\mathcal{L},\mathcal{S})$ of the stretchability problem is realizable if and only if its corresponding $(\mathcal{G},\mathcal{P},\mathcal{H})$ instance of the visibility graph is recognizable.
\end{theorem}
\begin{proof}
When $(\mathcal{L},\mathcal{S})$ is stretchable, we can obtain a polygon with holes from the realization of $(\mathcal{L},\mathcal{S})$ whose external and holes boundaries are respectively correspond to $\mathcal{P}$ and $\mathcal{H}$. On the other hand, when the channels are sufficiently narrow and close to their line segments, each pair of vertices see each other if and only if they belong to the same channel. This means that their visibility graph is $\mathcal{G}$ and this polygon with holes is a realization for $(\mathcal{G},\mathcal{P},\mathcal{H})$ instance.

To prove the theorem in reverse, assume that the instance $(\mathcal{G},\mathcal{P},\mathcal{H})$ is realizable and we have a polygon $\mathcal{P}$ with holes $\mathcal{H}$ whose visibility graph is $\mathcal{G}$. In this realization, for the pair of endpoints of each channel consider the line that connects this pair of points.
We claim that this set of lines is an answer for $(\mathcal{L},\mathcal{S})$ instance of the stretchability problem.

The induced subgraph of $\mathcal{G}$ on the vertices of a channel $C$ is a complete graph which implies that in the realization of $(\mathcal{G},\mathcal{P},\mathcal{H})$ these vertices must lie on the boundary of a convex polygon. On the other hand, each pair of channels $C_i$ and $C_j$ has exactly four vertices in common. Denote these common vertices by $Int(C_i,C_j)$. This set of common vertices forces that their corresponding lines (the lines that pass through the endpoints of these channels) must intersect at one point. Therefore, each pair of the obtained lines intersect at a point. To complete the proof, we must show that the intersections of these lines follow $\mathcal{S}$, and their initial vertical order is as $\mathcal{L}$. Consider the corresponding line $l$ of a channel $C$. The order of intersection points of $l$ and corresponding lines of other channels is directly derived from the order of common vertices between $C$ and other channels along the boundary of $C$. While these common vertices are either the vertices of $\mathcal{P}$ or cycles in $\mathcal{H}$, we can identify their order uniquely along the boundary of channel $C$. To do this, we first obtain the common vertices $Int(C,C_i)$, which have two vertices on $\mathcal{P}$. This pair of vertices are connected to an endpoint of the channel $C$. This means that $l$ is first intersected by the corresponding line of $C_i$. Let $p,q$ be the other vertices in $Int(C,C_i)$. These vertices either belong to $\mathcal{P}$ or a cycle in $\mathcal{H}$ and in both cases the next intersected channel is determined. When $p$ lies on $\mathcal{P}$, the next intersected channel is $C_j$ where $Int(C,C_j)$ has a vertex adjacent to $p$ in $\mathcal{P}$, and, when $p$ lies on a cycle $H \in \mathcal{H}$, the next intersected channel is $C_j$ where $Int(C,C_j)$ has a vertex adjacent to $p$ in $H$. Continuing this procedure, the order of intersection points of $l$ with other lines are uniquely determined which exactly is the same as $\mathcal{S}$. The reason is that we have built the channel and their common vertices according to $\mathcal{L}$. Finally, the initial order of these lines is the same as the order of endpoints of their corresponding channels along $\mathcal{P}$. Therefore, the initial vertical order of these lines will follow $\mathcal{L}$ by properly rotating the realization of the polygon for $(\mathcal{G},\mathcal{P},\mathcal{H})$.
\end{proof}

It is easy to show that recognizing the visibility graph of a polygon with holes belongs to $\exists \Real$. On the other hand, the stretchability problem is  $\exists \Real\text{-}Complete$, and our reduction is polynomial. Then, Theorem~\ref{thm:1_1} implies the $\exists \Real\text{-}Hardness$ of recognizing the visibility graph of a polygon with holes. Therefore, we have the following theorem.

\begin{theorem}
\label{thm: 1_2}
Recognizing the visibility graph of a polygon with holes is  $\exists \Real\text{-}Complete$.
\end{theorem}

\section{Complexity of Internal-External Visibility Graphs of a Simple Polygon}
\label{sec:D}
In this section, we show that recognizing the internal-external visibility graph of a simple polygon is $\exists\Real\text{-}Complete$. Again, we prove this by reducing an instance $(\mathcal{L},\mathcal{S})$ of the stretchability problem to an instance $(\mathcal{G}_{int},\mathcal{G}_{ext},\mathcal{P})$ of the visibility graph recognition where $\mathcal{G}_{int}$ and $\mathcal{G}_{ext}$ are respectively the internal and external visibility graphs and $\mathcal{P}$ is the external cycle of the boundary of the target simple polygon. Our construction of $(\mathcal{G}_{int},\mathcal{G}_{ext},\mathcal{P})$ is similar to the construction described in Section~\ref{sec:C}. We first build the same polygon with holes as in Section~\ref{sec:C} with this difference that on each segment $ab$ on the boundary of a hole in $\mathcal{H}$, which $a$ and $b$ are intersection points of two convex chains, another point $c$ is added(Fig.~\ref{fig:5}-a) without violating the convexity of the channel chains. Then, we connect the holes together and to the external area to remove the holes and obtain a simple polygon. This is done iteratively on each hole by adding a pair of parallel and sufficiently close \textit{cut edges} that connect a boundary segment of a hole to a segment of the current outer boundary. These segments must belong to different chains of the same channel and the pair of cut edges are close enough such that no pair $(p,q)$, in which $p$ and $q$ are not on the boundary of the same hole, be internally or externally visible to each other. These cut edges act as cutting channels like Fig.~\ref{fig:5}-b.

Trivially, the visibility graph of this simple polygon is no longer the same as the one obtained in Section~\ref{sec:C}. Having the stretched realization, we can compute internal and external visibility graphs of this simple polygon, which mainly depends on the way we cut the channels to remove the holes. In this polygon, two vertices are internally visible if and only if they belong to the same channel and their connecting segment does not cross any cut edge. In the external visibility graph, two vertices have an edge if and only if one of the following conditions holds:
\begin{itemize}
\item they are adjacent vertices on $\mathcal{P}$(points $a$ and $b$ in Fig.~\ref{fig:5}-b),
\item they are endpoints of the edges of the same cut(points $c$ and $d$ in Fig.~\ref{fig:5}-b),
\item they are vertices on different channel chains of the same hole (holes before cutting)(points $e$ and $f$ in Fig.~\ref{fig:5}-b),
\item they lie on different channel chains and between two consecutive endpoints of the channels on $\mathcal{P}$(before cutting but including cut vertices)(points $g$ and $h$ or $g$ and $i$).
\end{itemize}

\begin{figure}[ht]
\centerline{\includegraphics[scale=0.55]{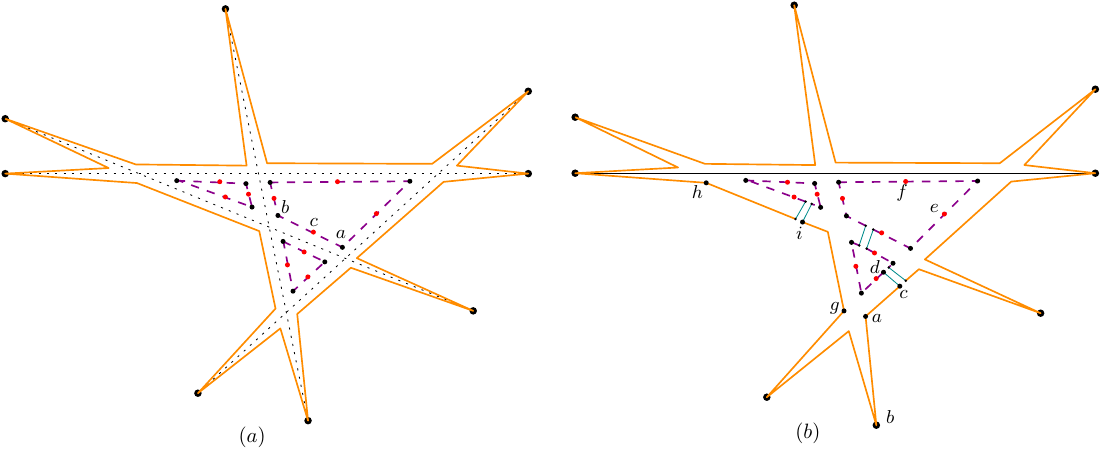}}
\caption{A simple polygon constructed from an instance of the \textit{PseudoLineArrangement} problem.\label{fig:5} }
\end{figure}

Trivially, if $(\mathcal{L},\mathcal{S})$ is stretchable, its corresponding instance $(\mathcal{G}_{int},\mathcal{G}_{ext},\mathcal{P})$ is recognizable. To prove the reverse equivalence, we show in the next theorem that in any realization of $(\mathcal{G}_{int},\mathcal{G}_{ext},\mathcal{P})$ the boundary vertices of each channel must be a convex polygon. Then, by an argument similar to the one in Theorem~\ref{thm:1_1}, the stretchability of $(\mathcal{L},\mathcal{S})$ is equivalent to the realization of $(\mathcal{G}_{int},\mathcal{G}_{ext},\mathcal{P})$ as a simple polygon.

\begin{theorem}
\label{thm:2_1}
In any realization of the internal-external visibility graph $(\mathcal{G}_int,\mathcal{G}_ext,\mathcal{P})$, the lower and upper chains of each channel $C$ are convex.
\end{theorem}
\begin{proof}
For the sake of a contradiction, assume that there is a channel $C$ in the realization of $(\mathcal{G}_{int},\mathcal{G}_{ext},\mathcal{P})$ that is not a convex polygon. Except for the vertices of the cut edges, for each vertex of $C$, its pair of adjacent vertices on this channel are visible from each other in $\mathcal{G}_{int}$. Therefore, such a concavity on the boundary of $C$ can occur only on some vertices of the cut edges that cross this channel. Without loss of generality, assume that such a concavity happens on the boundary of $C$ in Fig.~\ref{fig:6}, when we connect $O_1$ to $O_2$. The vertices $A_1$ and $O_2$ are non-visible pairs in $\mathcal{G}_{ext}$. According to Observation~\ref{cond:2}, they have at least one and at most two blocking vertices, one on the clockwise and the other on the counter-clockwise walks from $A_1$ to $O_2$. The vertex $O_1$ is such a blocking vertex on the counter-clockwise walk from $A_1$ to $O_2$. On the other hand, the last visible vertex from $A_1$ in $\mathcal{G}_{ext}$ along the clockwise walk on $\mathcal{P}$ toward $O_2$ is a vertex like $B_2$ and the last visible vertex from $O_2$ in $\mathcal{G}_{ext}$ along the counter-clockwise walk on $\mathcal{P}$ toward $A_1$ is another vertex like $B_1$ which is farther than $B_2$ from $O_2$. Note that $O_1$ and $O_2$ lie on the boundary of the same hole before cutting and we have added extra vertices on all boundary segments of this hole. The vertices $A_1$, $B_1$, and $B_2$ are such extra and distinct vertices. Then, according to Observation~\ref{obs:2}, the non-visible pair $(A_1,O_2)$ does not have a blocking vertex in the clockwise walk from $A_1$ to $O_2$. Therefore, in any realization of $(\mathcal{G}_{int},\mathcal{G}_{ext},\mathcal{P})$, $O_1$ is the only blocking vertex of $A_1$ and $O_2$ which means that $O_1$ lies above the segment $A_1O_2$ and vertex $O_1$ is a convex vertex on the boundary of $C$. The same argument implies that $O_2$ is also a convex vertex on the boundary of $C$, which contradicts our assumption that $C$ is concave on $O_1$ or $O_2$.
\end{proof}

\begin{figure}[ht]
\centerline{\includegraphics[scale=0.4]{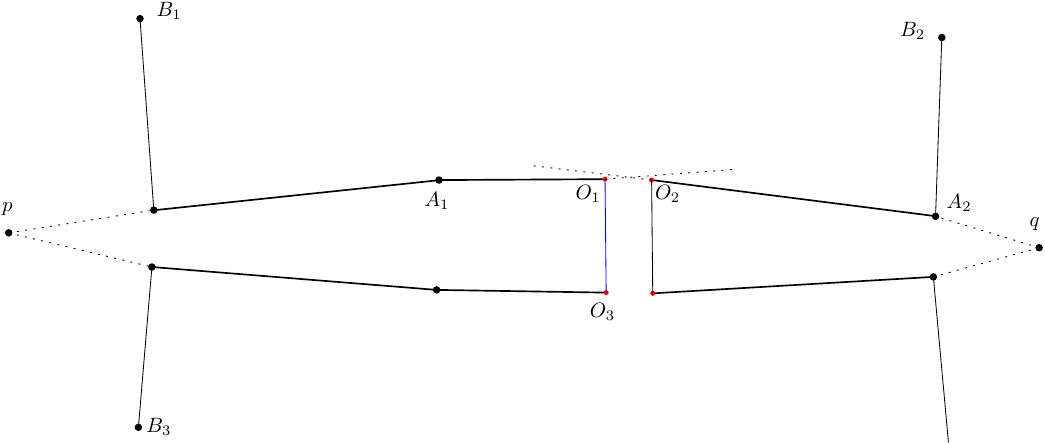}}
\caption{Upper and lower chains of a pseudo-line and their blocking vertices. \label{fig:6} }
\end{figure}

The above theorem implies that in any realization of $(\mathcal{G}_{int},\mathcal{G}_{ext},\mathcal{P})$, all channels are convex polygons. Therefore, by the same argument as Theorem~\ref{thm:1_1},  $(\mathcal{L},\mathcal{S})$ is stretchable if and only if $(\mathcal{G}_{int},\mathcal{G}_{ext},\mathcal{P})$ is recognizable. On the other hand, this reduction can be done in polynomial time and implies the $\exists \Real\text{-}Hardness$ of recognizing internal-external visibility graphs of a simple polygon. It is easy to show that recognizing internal-external visibility graphs of a simple polygon belongs to $\exists \Real$,  and we obtain the following theorem.

\begin{theorem}
\label{thm:2_2}
Recognizing internal-external visibility graphs of a simple polygon is $\exists \Real\text{-}Complete$.
\end{theorem}

\section{The Recognition Problem in Presence of a  \textit{BlockingVertexAssignment}}
\label{sec:BVA}
In this section, we discuss the effect of having a \textit{BlockingVertexAssignment} on the difficulty of the visibility graph recognition problem, \textit{i.e.} the problem can be solved easier when we have a \textit{BlockingVertexAssignment} as a part of the input. Recall that, \textit{BlockingVertexAssignment} of a visibility graph of a polygon is a function $\beta:V^2 \rightarrow V$ that assigns a vertex to each non-visible ordered pair of vertices. This function was introduced by Abello, \textit{et al.}, in 1994~\cite{matroid} and indicates which vertex is the first one that blocks the visibility of an ordered pair of vertices. This blocking vertex is visible to the first vertex in the ordered non-visible pair and it is clear that every non-visible pair must have such a blocking vertex. There is another way to define a blocking vertex of a non-visible ordered pair: it is the first vertex in the shortest Euclidean path inside the polygon from the first to the second vertex of a non-visible ordered pair. Such a path is unique inside a simple polygon, but, there can be more than one such path between a pair of vertices of a polygon with holes. Moreover, two vertices on the boundary of a simple polygon can have two different shortest paths that lie outside the polygon. Therefore, \textit{BlockingVertexAssignment} is not well-defined for the visibility graph of a polygon with holes and for the external visibility graph of a simple polygon. Abello \textit{et al.}, showed that a \textit{BlockingVertexAssignment} must have four conditions, which are necessary conditions for a visibility graph to be recognizable. Later in 1997, Ghosh showed that these conditions are strictly stronger than Ghosh's first three necessary conditions, but, strictly weaker than Ghosh's four necessary conditions~\cite{ghoshn}. In 1995, Abello \textit{et al.} introduced a polynomial time algorithm that recognizes the visibility graph of a 2-spiral polygon from its visibility graph when we have \textit{BlockingVertexAssignment} as part of the input~\cite{twospiral}. This result initiated the question of whether having such a \textit{BlockingVertexAssignment} function as a part of the input makes the recognizing visibility graphs problem easier (in terms of computational complexity) or not. In this section, we show that recognizing the visibility graph problem for simple polygons is still $\exists \Real\text{-}Complete$ when we have both internal and external visibility graphs and a \textit{BlockingVertexAssignment} function for its internal visibility graph as input.

\begin{theorem}
\label{thm:3}
The internal visibility graph of each polygon constructed from the \textit{PseudoLineArrangement} problem has a unique \textit{BlockingVertexAssignment} function and can be computed in polynomial time.
\end{theorem}
\begin{proof}
\label{proof_thm:3}
Consider a non-visible ordered pair $(a,b)$ and a blocking vertex assignment function $\mathcal{B}$, where $a$ belongs to a channel $C$. Like any non-visible ordered pair, there are two candidates for $\mathcal{B}(a,b)$: the last visible vertex in clockwise (name it $w$) and counter-clockwise (name it $c$) walks from $a$ toward $b$ on the boundary of the polygon. Both $w$ and $c$ are visible from $a$. Ghosh showed that these two blocking vertex candidates see each other~\cite{ghoshn}. Therefore, $a$, $c$ and $w$ see each other and belong to the channel $C$. 

On the other hand, $c$ and $w$ do not belong to the same common vertices of $Int(C,C_j)$. For each pair of channels $(C_i, C_j)$, all vertices of $Int(C_i,C_j)$ are visible from the previous and next vertices adjacent to the other three vertices in $Int(C_i,C_j)$. Therefore, by using Observation~\ref{obs:2}, if more than one of these three vertices belongs to the clockwise (resp. counter-clockwise) walk from $a$ toward $b$, $(a,b)$ has no blocking vertex in clockwise (resp. counter-clockwise) walk from $a$ toward $b$.

Let's denote the three vertices which are in the same intersection set with $c$ by $o_1$, $o_2$ and $o_3$ and those which are in the same intersection set with $w$ by $o^\prime_1$, $o^\prime_2$ and $o^\prime_3$. If both $c$ and $w$ are blocking vertices, $o_1$, $o_2$ and $o_3$ belong to the clockwise walk from $a$ toward $b$ and in this walk, we must meet them before $w$. Similarly, we must meet $o^\prime_1$, $o^\prime_2$ and $o^\prime_3$ before $c$ in the counter-clockwise walk from $a$ toward $b$. Hence, vertices $a$, $b$, $c$, and $w$ have an arrangement similar to the arrangement shown in Fig.~\ref{fig:7}. As shown in this figure, we cannot cross the convex area surrounded by vertices of the channel $C$ which are visible from $a$ (like the area surrounded by dotted orange, blue, and green curves in Fig.~\ref{fig:7}). In all arrangements of these points, some of the adjacent vertices of at least one of these eight intersection points will be trapped between the orange curve, the clockwise walk from $w$ toward $b$, and the counter-clockwise walk from $c$ toward $b$ (see $x$ and $y$ in Fig.~\ref{fig:7}). Therefore, we cannot meet these vertices before $w$ or $c$ in a clockwise or counter-clockwise walk from $a$ toward $b$ without crossing the aforementioned convex area. Consequently, we cannot have more than one visible blocking vertex from $a$ for the ordered non-visible pair $(a,b)$. Hence, the \textit{BlockingVertexAssignment} for this visibility graph is unique.

\begin{figure}[ht]
\centerline{\includegraphics[scale=0.9]{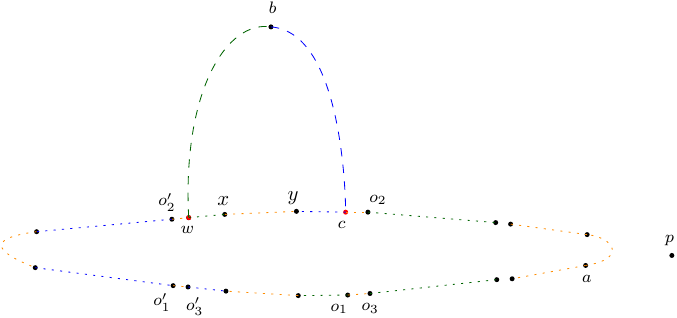}}
\caption{Determining \textit{BlockingVertexAssignment}: The blue dotted path corresponds to a part of the counter-clockwise walk from $a$ toward $c$ and the green dotted path corresponds to a part of the clockwise walk from $a$ toward $w$. The blue dashed curve corresponds to the counter-clockwise curve from $c$ toward $b$ and the green dashed curve corresponds to the clockwise curve from $w$ toward $b$. The orange curve corresponds to the convex area which is visible from $a$.\label{fig:7} }
\end{figure}

We only use the first two of Ghosh's necessary conditions to find blocking vertices and these conditions can be verified in polynomial time~\cite{ghoshn}. So we can compute the \textit{BlockingVertexAssignment} for this kind of polygons in polynomial time. 
\end{proof}

This theorem along with Theorem~\ref{thm:2_2} implies that if we have both internal and external visibility graphs of a simple polygon and its \textit{BlockingVertexAssignment}, the recognition problem is still $\exists \Real\text{-}Complete$.

\begin{theorem}
Recognizing the internal-external visibility graph of a simple polygon is $\exists \Real\text{-}Complete$ even if we have a \textit{BlockingVertexAssignment} function other than both internal and external visibility graphs.
\end{theorem}

\section{Conclusion}
In this paper, we showed that the visibility graph recognition problem is $\exists \Real\text{-}Complete$ for polygons with holes. Moreover, we showed that having both internal and external visibility graphs of a simple polygon, it is $\exists \Real\text{-}Complete$ to recognize them as corresponding visibility graphs of a simple polygon. Although the actual complexity class of recognizing visibility graphs for simple polygons is still open, from this result we conjecture that the problem is $\exists \Real\text{-}Complete$. Note that having the external visibility graph reduces possible options for reconstructing a simple polygon of a given visibility graph which offers more hints to the reconstruction procedure (makes it easier) on one side and adds more constraints (makes it more difficult) on the other side. So, we cannot definitely determine the complexity class of the problem when we have only the internal visibility graph compared to the cases when we have the external visibility graph as well 

We also proved that having a \textit{BlockingVertexAssignment} for the internal visibility graph as part of the input does not reduce the complexity class of this problem. However, there is another long-standing open problem about the complexity of recognizing a simple polygon when we have a \textit{BlockingVertexAssignment} for its non-visible pairs of vertices. For the same reason, this result also motivates us to guess that this problem is $\exists \Real\text{-}Complete$ as well.

%\bibliographystyle{abbrvnat}
% use the following instead if you encounter problems 
\bibliographystyle{alpha}
\bibliography{dmtcs}
\label{sec:biblio}

\end{document}